\documentclass[showkeys,floatfix,noshowpacs, preprintnumbers, twocolumn, aps,superscriptaddress,pra,10pt]{revtex4-1}

\usepackage{amsmath,amsthm,amssymb,amsfonts}

\usepackage{graphicx}
\usepackage{multirow}
\usepackage{url}
\usepackage[usenames,dvipsnames]{xcolor}
\usepackage[breaklinks=true]{hyperref}
\usepackage{bbold}
\usepackage{cleveref} 
\usepackage[normalem]{ulem}

\crefformat{equation}{Eq.~(#2#1#3)} 
\Crefformat{equation}{Equation~(#2#1#3)}
\crefformat{figure}{Fig.~#2#1#3}
\crefrangeformat{equation}{Eqs.~#3(#1)#4--#5(#2)#6}

\newtheorem{lemma}{Lemma}

\newcommand{\bra}[1]{\left\langle{#1}\right\vert}
\newcommand{\ket}[1]{\left\vert{#1}\right\rangle}
\newcommand{\ketbra}[2]{\left\vert{#1}\right\rangle\!\left\langle{#2}\right\vert}
\newcommand{\braket}[2]{\left\langle{#1}\vert{#2}\right\rangle}
\newcommand{\ovlap}[2]{\left \lvert\braket{#1}{#2}\right\lvert^2}
\newcommand{\ovlapr}[2]{r_{#1 #2}}

\newcommand{\tr}{\mathrm{tr}}

\begin{document}

\title{Quantum and classical bounds for two-state overlaps}

\author{Ernesto F. Galv\~{a}o}
\affiliation{Instituto de F\'{i}sica, Universidade Federal Fluminense, Av. Gal. Milton Tavares de Souza s/n, Niter\'oi, RJ, 24210-340, Brazil}
\affiliation{International Iberian Nanotechnology Laboratory (INL), Av. Mestre Jos\'{e} Veiga, 4715-330 Braga, Portugal}

\author{Daniel J. Brod}
\affiliation{Instituto de F\'{i}sica, Universidade Federal Fluminense, Av. Gal. Milton Tavares de Souza s/n, Niter\'oi, RJ, 24210-340, Brazil}

\begin{abstract}
Suppose we have $N$ quantum systems in unknown states $\ket{\psi_i}$, but know the value of some pairwise overlaps $\left| \braket{\psi_k}{\psi_l}\right|^2$. What can we say about the values of the unknown overlaps? We provide a complete answer to this problem for three pure states and two given overlaps, and a way to obtain bounds for the general case. We discuss how the answer contrasts from that of a classical model featuring only coherence-free, diagonal states, and describe three applications: basis-independent coherence witnesses, dimension witnesses, and characterisation of multi-photon indistinguishability. 
\end{abstract}

\maketitle

\section{Introduction}
Equality is an example of an equivalence relation. In particular, it is transitive, so $\forall A, \forall B, \forall C$, ($A=B$ $\wedge$ $A=C) \implies B=C$. This allows us to infer the equality of two objects without ever directly comparing them. In this work, we will be interested in inferring as much as possible about two-state comparisons that were never made, based on information on two-state comparisons that were actually performed. There are various notions of quantum state comparison \cite{Audenaert14}. Here we focus on the two-state overlap (or linear fidelity) $r_{AB}=\tr(\rho_A \rho_B)$, which for pure states reduces to $r_{AB}= \left| \braket{A}{B} \right|^2$. Due to the probabilistic nature of quantum theory, a natural question is: given the values of the pairwise overlaps between some pairs of states in a set, what can we establish about the unknown pairwise overlaps in the set?

Here we provide a complete answer to this problem for three pure states, and show how to obtain bounds for unknown overlaps in the general scenario.  Our results describe a fundamental aspect of the geometry of allowed quantum states, by characterizing the constraints satisfied by the degrees of similarity between pairs of states in a set. We will see that our results allow for experimentally simpler designs for the characterization of multiphoton indistinguishability, a resource for new photonic quantum technologies \cite{ObrienFV09}.

If instead of general quantum states, we consider some subset of states, the bounds on overlaps can be more restrictive. We find the bounds satisfied by two different subsets of quantum states. First, we characterize the bounds that apply to three classical states, defined as coherence-free states which are simultaneously diagonal in some single reference basis. We also consider the bounds that result if we assume that the three states span only a 2-dimensional Hilbert space, as will be the case, e.g., when we have three states of single qubits. By contrasting these more restrictive bounds with the general bounds on overlaps we obtained, we present two other applications of our results: basis-independent coherence witnesses \cite{StreltsovAP17}, and Hilbert-space dimension witnesses \cite{BrunnerPAGMS08}.

Our paper is organized as follows. In Section \ref{sec:3qs} we obtain the bounds that apply to a scenario with three quantum states, and show how to generalize this to more general scenarios with any number of states. In Section \ref{sec:cbounds} we obtain the bounds that apply to classical, coherence-free states which are diagonal in some common basis. In Section \ref{sec:applic} we describe three applications of our results: basis-independent coherence witnesses, characterization of multiphoton indistinguishability, and Hilbert space dimension witnesses. We offer some concluding remarks in Section \ref{sec:conclusion}, with some of the more technical calculations presented in the Appendices.

\section{Overlap bounds for quantum states}
\label{sec:3qs}
Let us start by formalising the general problem we consider and introducing notation. We refer to the \emph{overlap} between two quantum states $\rho_i$ and $\rho_j$ as $r_{ij} = \tr (\rho_i \rho_j)$. Unless stated otherwise, in this paper we assume that states are pure, in which case the overlap reduces to
\begin{equation}
r_{ij}=\left|\braket{\psi_i}{\psi_j}\right|^2.
\end{equation}

Let $G$ be an $N$-vertex weighted connected graph. Let vertices represent unknown quantum states $\ket{\psi_i}$, with edges representing known two-system overlaps with weights given by  $r_{ij}$ (see \cref{fig:3chain}). The problem is to obtain tight bounds for the possible values for the unknown overlaps or, more generally, to characterize the geometry of the space of allowed pairwise overlaps.

\begin{figure}[b]
    \centering
    \includegraphics{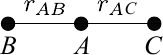}
    \caption{$P_3$ graph. Nodes represent pure states, edges represent known pairwise overlaps $r_{AB}=\left|\braket{A}{B}\right|^2$ and $r_{AC}=\left|\braket{A}{C}\right|^2$.}
    \label{fig:3chain}
\end{figure}

One way to measure the overlap between the states of two quantum systems $A$ and $B$, of arbitrary dimensions, without obtaining any other information, is to use the well-known SWAP test \cite{BuhrmanCWdW01, MontanarodW16}. The SWAP test, represented in \cref{fig:swap2}(a), consists of a projection onto the symmetric subspace. As long as $A$ and $B$ are in a tensor product state, the probability of obtaining measurement outcome $0$ from the top auxiliary qubit in \cref{fig:swap2}(a) is $p(0)=( 1+r_{AB})/2$, which thus allows us to estimate the overlap $r_{AB} = \tr(\rho_A \rho_B)$.
If the two systems are qubits, \cref{fig:swap2}(b) shows a simplified version of the SWAP test, consisting of a measurement onto the Bell basis. In this case, the final measurements yield $p(11)=(1-r_{AB})/2$. A similar simplified implementation of $n$-qubit SWAP tests is presented in \cite{Garcia-EscartinCP13}. The SWAP test has applications in quantum fingerprinting \cite{BuhrmanCWdW01}, quantum machine learning \cite{BiamonteWPRWL17}, and in the characterisation of monogamy relations between fermionic and bosonic behaviors \cite{KarczewskiKK18}.

\begin{figure}[t]
    \centering
    \includegraphics[width=0.40\textwidth]{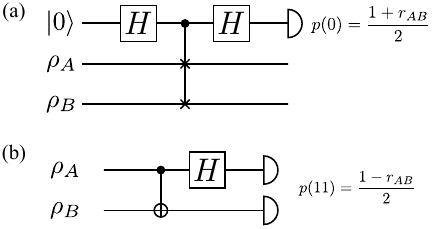}
    \caption{(a) SWAP test, allowing to estimate two-state overlap $r_{AB}=\tr(\rho_A \rho_B)$ of states of any dimensionality. (b) Simplified SWAP test for two states of one qubit. }
    \label{fig:swap2}
\end{figure}

\subsection{Bounds for overlaps of 3 states} \label{sec:b3states}

We begin by considering this problem for the smallest case of $N=3$ pure states. Later we will discuss how to leverage the solution of this case to obtain non-trivial bounds for scenarios with $N>3$ pure states.

The only non-complete connected graph with 3 vertices is $P_3$, the 3-vertex chain graph (see Fig.\ \ref{fig:3chain}). Let us call the two vertices at the ends of the chain $B$ and $C$, with $A$ being the degree-2 vertex. Suppose we know the values of 
\begin{align*}
r_{AB} &= \left| \braket{A}{B}\right|^2 ,\\
r_{AC} &= \left| \braket{A}{C}\right|^2.
\end{align*}
Our goal is to obtain non-trivial bounds for the unmeasured overlap $r_{BC}= \left|\braket{B}{C}\right|^2$, by varying over $\{ \ket{A}, \ket{B}, \ket{C}\}$ with fixed $r_{AB}$ and $r_{AC}$. 

Any three pure states span a Hilbert space which is at most 3-dimensional. Without loss of generality we choose a basis $\{ \ket{0}, \ket{1}, \ket{2}\}$  such that:
\begin{align*}
\ket{A}=&\ket{0},\\
\ket{B}=&\cos\beta\ket{0}+\sin\beta\ket{1},\\
\ket{C}=&\cos\gamma\ket{0}+\sin\gamma\sin\alpha e^{i \phi}\ket{1}\nonumber \\ &+\sin\gamma\cos\alpha\ket{2}.
\end{align*}
with $\alpha, \beta, \gamma \in [0, \pi/2]$ and $\phi \in [0,2\pi)$. This is possible by choosing the arbitrary global phases of each state, and applying diagonal unitaries in the computational basis to eliminate one relative phase from both $\ket{B}$ and $\ket{C}$.

Now we can extremise
\begin{equation*}
r_{BC}=\left|\cos \gamma\cos\beta+e^{i \phi}\sin\gamma\sin\beta\sin\alpha \right|^2 ,
\end{equation*}
subject to the constraints of fixed $r_{AB}=\cos^2 \beta$ and $r_{AC}=\cos^2\gamma$. This is done in Appendix \ref{apx:overlaps}, and the result is as follows. To write these bounds explicitly in terms of the given overlaps we first define the shorthand
\begin{equation*}
r_{\pm} := \left( \sqrt{r_{AB}r_{AC}}\pm\sqrt{(1-r_{AB})(1-r_{AC})}\right)^2
\end{equation*}
in terms of which we can write
\begin{equation}
r_{BC} \le r_{+} \label{eq:ub}
\end{equation}
and
\begin{equation}
  r_{BC}\ge \begin{cases}
    r_{-}, & \text{if $r_{AB}+r_{AC} > 1$,}\\
    0 & \text{otherwise}.
  \end{cases}\label{eq:lb}
\end{equation}
In all cases the bounds are tight, in the sense that there always exist triples of pure states for which they are attained. 

We stress that the bounds (\ref{eq:ub})-(\ref{eq:lb}) apply to three pure quantum states of \emph{arbitrary} dimension. In Appendix \ref{apx:qubits} we prove that they also hold for three \emph{mixed} single-qubit states. We conjecture that the bound holds in full generality, i.e.\  for mixed states of any dimension.

Besides bounding one overlap in terms of the others, we can also view our results from the perspective of the geometry of the set of allowed overlap triples. More specifically, if we consider $\vec{r}=(r_{AB}, r_{AC}, r_{BC})$ as a real triplet, then the set  $Q$ of allowed values of $\vec{r}$ compatible with three pure quantum states is characterized by the conditions that $r_{ij} \in [0,1]$, together with the inequality
\begin{equation}
   r_{AB}+r_{BC}+r_{AC}-2\sqrt{r_{AB}r_{BC}r_{AC}}\leq1, \label{eq:implic3s} 
\end{equation}
which is obtained by rearranging the terms in our upper and lower bounds. The set $Q$ is pictured in \cref{fig:ditqudit}(a).

\begin{figure*}[ht]
    \centering
    \includegraphics[width=0.80\textwidth]{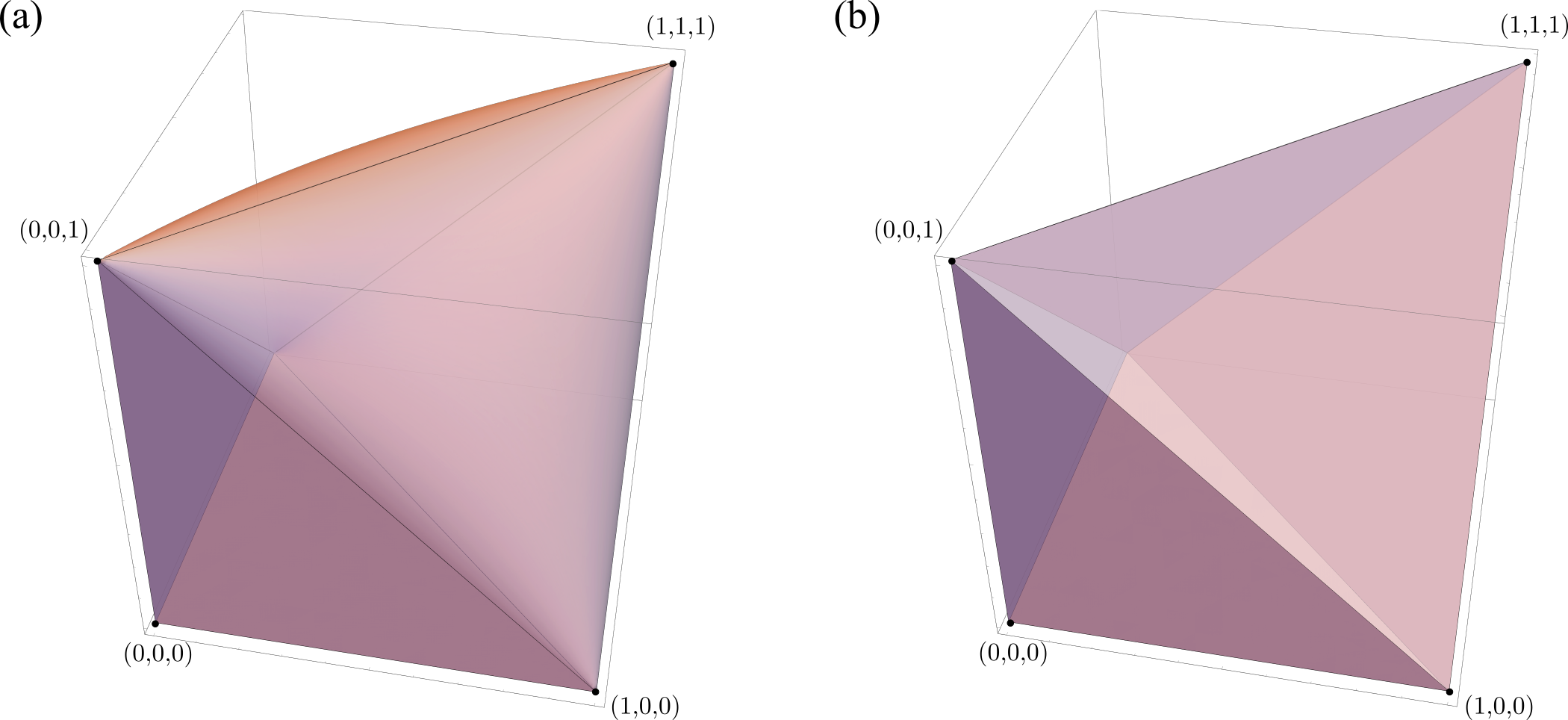}
    \caption{\textbf{Classical and quantum bounds for 3 overlaps of 3 states}. These are represented in the space of overlap triples $\vec{r}=(r_{AB}, r_{AC}, r_{BC})$. (a) Quantum convex body $Q$  corresponds to overlaps attainable by general quantum states. $Q$ can be described compactly by inequality $r_{AB}+r_{BC}+r_{AC}-2\sqrt{r_{AB}r_{BC}r_{AC}}\leq1$.
    (b) Classical polyhedron $C$ is the convex hull of vertices $\{ (0,0,0), (0,0,1), (0,1,0), (1,0,0), (1,1,1)\}$, and describes the region attainable by three states which are diagonal in a single basis. $Q$ properly contains $C$, but also some bulging regions off each non-trivial face of $C$.}
    \label{fig:ditqudit}
\end{figure*}

\subsection{Bounds on overlaps of more than 3 states} \label{sec:graphs}

The bounds (\ref{eq:ub})-(\ref{eq:lb}) obtained for the $P_3$ graph can be leveraged to obtain bounds for quantum states and overlaps described by a general connected graph $G$ with any number $N>2$ of vertices (systems) and edges (known two-state overlaps). We can do this by decomposing $G$ into $P_3$ subgraphs and repeatedly applying the 3-state bounds, as follows.

Consider one $P_3$ subgraph of $G$, labelling its vertices as $A$, $B$ and $C$ as before. By applying our previous reasoning, we obtain an interval of possible (and attainable) values for $r_{BC}$. 
We then cycle $A$ over all vertices that are adjacent to both $B$ and $C$ in $G$, and take the intersection of the inferred possible ranges for $r_{BC}$. A new bona fide edge is then added between vertices $B$ and $C$, with weight given by the corresponding intersection of ranges for $r_{BC}$, creating a new graph $G'$.

The above procedure can be repeated by taking, at every step, a new pair of vertices with a common neighbor and adding a new edge between them, until we have a complete graph with weights corresponding to inferred overlap ranges. In the intermediate steps, one or both overlaps $r_{AB}$ or $r_{AC}$ might actually correspond to a range inferred in a previous step. In that case, the range of values for $r_{BC}$ will be the union of all values allowed by all possible values of $r_{AB}$ or $r_{AC}$ in their corresponding inferred ranges. This procedure is illustrated in Fig.\ \ref{fig:penta} for a 5-cycle graph.

\begin{figure}[t]
    \includegraphics[width=0.45\textwidth]{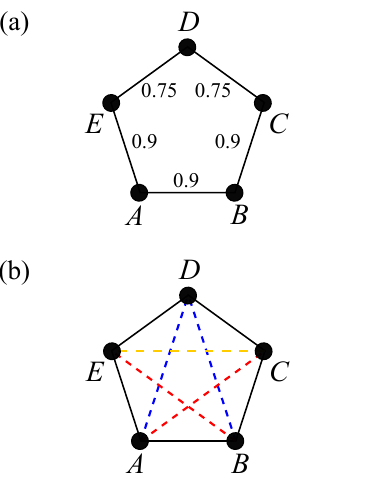}
    \caption{ (a) A 5-cycle graph of given overlaps. Vertices represent pure states and solid edges represent given pairwise overlaps, with values as in the figure. (b) We apply the 3-state bounds in the main text on all $P_3$ subgraphs of the 5-cycle to infer ranges of values for the dashed edges. We find that the red edges $r_{BE},r_{AC} \in [0.64,1]$, the blue edges $r_{AD},r_{BD} \in [0.44, 0.96]$, and the yellow edge $r_{CE} \in [0.25, 1]$. These bounds can, however, be tightened by revisiting the yellow edge, bounding it by a 3-state argument using the original edge $r_{BC}=0.9$ together with  edge $r_{BE} \in [0.64,1]$ inferred in the first round, to conclude that   $r_{CE} \in [0.324,1] $.
   }
    \label{fig:penta}
\end{figure}

After obtaining a complete graph, it may be necessary to revisit each inferred edge to tighten its bounds using edges that were added after it in the sequence of steps (this is also observed in the example of Fig.\ \ref{fig:penta}). We leave it as an open question whether this procedure is optimal, but emphasize that even the initial iteration that produces a complete graph already provides bounds for every pair of vertices, which may be useful even if not tight.

If one is interested only in tight \emph{lower} bounds, the procedure above becomes simpler. It can be shown that the lower bound (\ref{eq:lb}) is a monotonically increasing function of $r_{AB}$ and $r_{AC}$. Therefore, to obtain a lower bound for $r_{BC}$ in a step where $r_{AB}$ and/or $r_{AC}$ are already given by ranges inferred in previous steps, it suffices to take the minimum of each range.

We have completely characterized the two-state overlap values achievable by three pure quantum mechanical states of any dimension, and shown how to obtain bounds for the case of an arbitrary number of states. In the next section we derive more stringent bounds satisfied by the overlaps of coherence-free, diagonal quantum states. These bounds will be crucial for the basis-independent coherence witnesses we propose in section \ref{sec:cohe}.

\section{Overlap bounds for coherence-free states} \label{sec:cbounds}

We now derive a classical analogue of the bounds obtained in the previous section which must be satisfied by triplets of coherence-free states, i.e., states which are diagonal in some common basis.

To be precise, suppose we have some reference observable $\hat{O}$, with its orthonormal basis of eigenvectors $\{ \ket{\phi_i} \}$. If a state $\rho$  is diagonal in this basis, i.e.,
\begin{equation}
  \rho = \sum_i p_i \ketbra{\phi_i}{\phi_i}, \label{eq:cs}
\end{equation}
this state is  coherence-free with respect to $\hat{O}$. If we measure $\hat{O}$ on a coherence-free state such as (\ref{eq:cs}), we directly recover the classical probability distribution $\{p_i\}$.

Consider now two states, $\rho$ and $\sigma$, which are coherence-free with respect to a common observable $\hat{O}$. Their overlap, $\tr(\rho \sigma)$, can be interpreted as the probability of obtaining the same outcome $v(O)$ when measuring $\hat{O}$ independently on the two states:
\begin{align*}
\tr(\rho \sigma) =& \sum_i \bra{\phi_i} \rho \sigma \ket{\phi_i}\\
=&\sum_i \bra{\phi_i} \rho \ket{\phi_i}\bra{\phi_i}\sigma \ket{\phi_i}\\
=& \text{probability that } v(\hat{O})_{\rho}=v(\hat{O})_{\sigma}.
\end{align*}
Note that the overlap encodes some information about the underlying \emph{classical} probability distributions of the two coherence-free states. Besides being basis-independent, the overlap is a quantity that is operationally well-defined, via the SWAP test. In what follows, we prove bounds analogous to inequalities (\ref{eq:ub})-(\ref{eq:lb}), but which must be satisfied for any triplet of classical probability distributions (hence we refer to them as classical bounds). It will immediately follow that the bounds also hold for arbitrary triplets of coherence-free states---i.e., for three states which are diagonal in \emph{some} common, but possibly unknown, basis.

Consider now three independent classical processes $A, B$, and $C$ which assume values $v(A), v(B)$ and $v(C)$ with respective probabilities $p[v(A)], p[v(B)]$ and $p[v(C)]$. Let 
\begin{equation}
    p_{XY}:=p[v(X)=v(Y)]
\end{equation}
denote the probability that the drawn values for two processes $X$ and $Y$ coincide. 

We now use the fact that, given two propositions $a_1, a_2$, it is always true that $p(a_1 \wedge a_2) = p(a_1)+p(a_2)-p(a_1 \vee a_2)$, and consequently
\begin{equation}
p(a_1 \wedge a_2) \ge p(a_1)+p(a_2)-1. \label{lc2}
\end{equation}
If we assign propositions $a_1$ and $a_2$ as follows:
\begin{align*}
a_1 &:= v(A)=v(B)\\
a_2 &:= v(A)=v(C),
\end{align*}
inequality (\ref{lc2}) yields:
\begin{align*}
p[v(A)=v(B) \wedge v(A)=v(C)] \ge p_{AB}+p_{AC}-1.
\end{align*}
Since 
\begin{equation*}
    p_{BC} \ge p[v(A)=v(B) \wedge v(A)=v(C)],
\end{equation*}
we have
\begin{equation} \label{eq:lci1}
p_{BC} \ge p_{AB} + p_{AC}-1.  
\end{equation}

By permuting the indices in inequality (\ref{eq:lci1}), we additionally obtain the following:
\begin{align}
p_{BC} &\le p_{AB}-p_{AC}+1. \label{eq:lci2}\\
p_{BC} &\le p_{AC}-p_{AB}+1. \label{eq:lci3}
\end{align}
Ref.\ \cite{Pitowsky94} gives an introduction to such linear inequalities describing logical coherence, first discussed by George Boole in 1854 \cite{Boole1854}.

We now use the fact that, for states which are coherence-free in a common basis, $p_{AB}=\tr(\rho_A \rho_B)=r_{AB}$ coincides with the overlap. This yields the following classical inequalities for the three two-state overlaps $r_{ij}=\tr(\rho_i \rho_j)$:
\begin{align}
r_{BC} &\ge r_{AB}+r_{AC}-1, \label{eq:rineq1}\\
r_{BC} &\le r_{AB}-r_{AC}+1, \label{eq:rineq2}\\
r_{BC} &\le r_{AC}-r_{AB}+1. \label{eq:rineq3}
\end{align}

Inequalities (\ref{eq:rineq1})--(\ref{eq:rineq3}) are classical analogues of (\ref{eq:ub})--(\ref{eq:lb}) obtained previously for pure quantum states, and define a corresponding region $C$ in the space of allowed overlaps, represented in \cref{fig:ditqudit}(b). 
We reiterate that these inequalities must be satisfied for any set of three states diagonal in the same basis, but that they are given in terms of basis-independent operational quantities. In particular, this means that we can leverage violations of these inequalities to signal the \emph{inexistence} of a common basis in which the three states are coherence-free,  as we discuss in the next section.

As done in \cref{sec:graphs}, we can also consider extending inequalities (\ref{eq:rineq1})--(\ref{eq:rineq3}) to scenarios involving more states and known overlaps. In Appendix \ref{ap:classb} we extend this argument to obtain the following classical bound that applies to any $m$-edge connected graph $G$:
\begin{equation}
r_{kl} \ge 1-m+\sum_{\{i,j\} \in G} r_{ij}, \label{eq:genineq}
\end{equation}
where $k,l$ are any two vertices of $G$, not necessarily connected by an edge. For any pair $\{k,l\}$, the expression above actually represents multiple inequalities arising from all possible connected subgraphs of $G$ that contain those two vertices.

\section{Applications}
\label{sec:applic}
In this section we discuss three applications of our results. In section \ref{sec:cohe} we show how our classical inequalities can be used as a novel type of basis-independent coherence witness. In section \ref{sec:multi} we describe how to obtain experimentally-friendly ways to characterize multiphoton indistinguishability. Finally, in section \ref{sec:dim} we describe how overlap measurements can work as a Hilbert-space dimension witness.

\begin{figure*}[ht]
    \centering
    \includegraphics[width=0.8\textwidth]{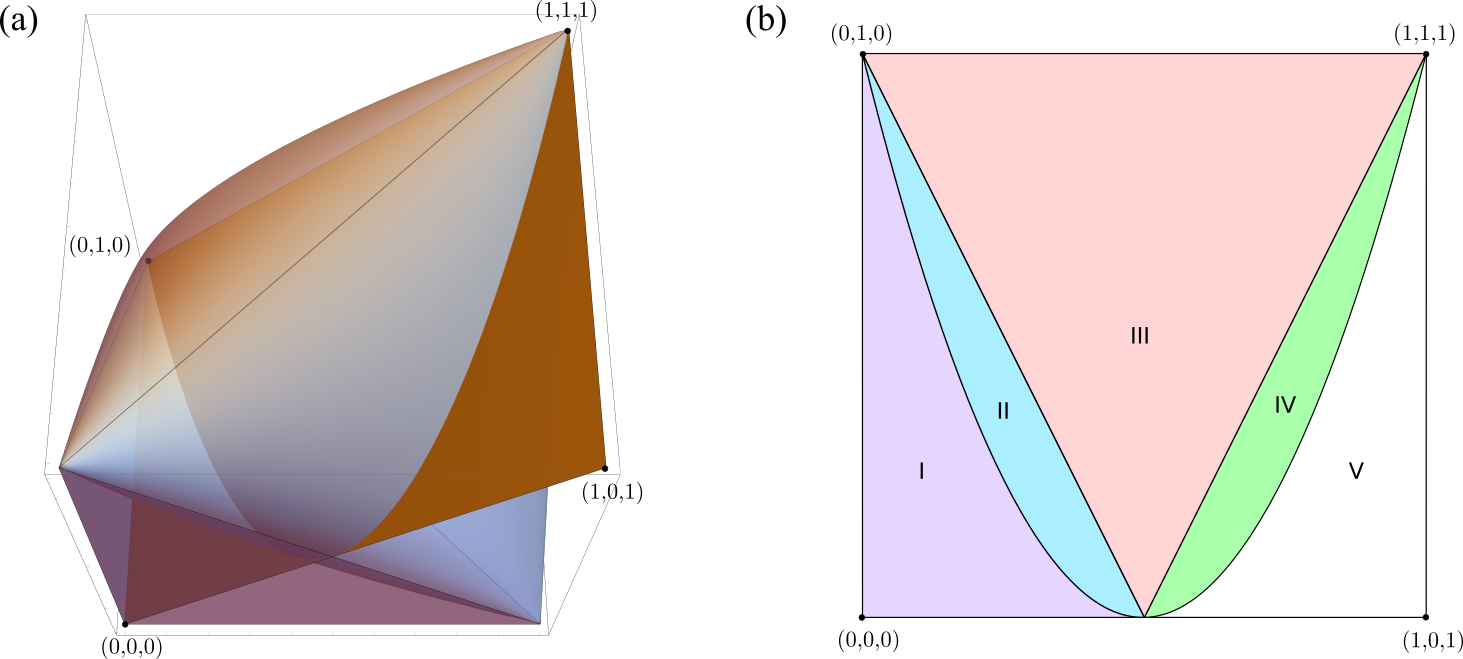}
    \caption{\textbf{Coherence and dimension witnesses in the space of overlap triples $\vec{r}=(r_{AB}, r_{AC}, r_{BC})$.} (a) We take a cross section of the polyhedron of classical, coherence-free states $C$, and its enveloping convex body of general quantum states $Q$, as indicated in the figure. (b) Regions in the cross section plane. Regions I, II, III are in $C$. Overlaps in region I are dimension witnesses, as they are only attainable by (classical or quantum) states spanning a Hilbert-space of dimension $d > 2$. Region II is reached either by quantum states spanning $\ge 2$ dimensions, or by classical states spanning 3 dimensions. Region III is reached by classical or quantum states spanning $\ge 2$ dimensions. Overlaps in region IV are coherence witnesses, as these regions are in $Q$ but not in $C$. Region V cannot be reached by quantum or classical overlaps, as it lies outside of $Q$.
    }
    \label{fig:regions}
\end{figure*}

\subsection{Basis-independent coherence witnesses} \label{sec:cohe}

Quantum coherence has been identified as a resource for quantum algorithms, quantum biology, and quantum thermodynamics \cite{StreltsovAP17}. One way to identify coherence in quantum states is to define a reference basis for the Hilbert space, and then find a coherence witness $W$ \cite{NapoliBCPJA16}, which is an observable whose expectation value can only be negative for states which are not diagonal in the reference basis.

Our results enable us to define a novel type of coherence witness that identifies superposition states even in the absence of information about the reference basis. Our new coherence witnesses require sources of three or more states, and the  measurement of a set of relevant two-state overlaps. If the witness is violated, we know that there is no single basis in which all states are diagonal. Note that this type of witness cannot work for just two states, as the whole range for a single overlap $r_{ij} \in [0,1]$ can be attained by mixtures of states in a fixed basis $\{\ket{i}, \ket{j}\}$. The minimum number of measured overlaps that is sufficient to establish a coherence witness is three, when we have three states and measure all  pairwise overlaps.

It is possible to construct a single quantum circuit that directly measures all overlaps,
$r_{AB}, r_{AC}$ and $r_{BC}$, of three quantum states $\rho_A, \rho_{B}$ and $\rho_{C}$. This circuit is a generalization of the SWAP test, and we represent it in \cref{fig:3swaps}. The probability of obtaining outcome 0 on the first qubit is $p(0)=(1+r_{ij})/2$, as discussed when we reviewed the standard SWAP test. The two first controlled-SWAP gates perform permutations that enable measurement of all pairwise overlaps, and the initial Hadamards prepare a uniform computational basis superposition of the two bottom qubits. Observing outcomes $a_0, a_1 =00, 01, 10, 11$ respectively herald that overlaps $r_{AB}, r_{BC}, r_{AC}, r_{BC}$ were measured.

\begin{figure}[ht]
    \centering
    \includegraphics[width=0.40\textwidth]{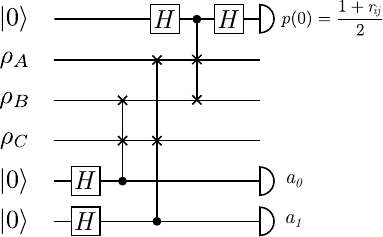}
    \caption{Circuit to measure the three pairwise overlaps of three quantum states $\rho_{A},\rho_{B}$ and $\rho_{C}$. The top auxiliary qubit acts as the control of a standard SWAP test [see \cref{fig:swap2}(a)] , which allows an estimate of the overlap. The first two controlled-SWAP gates serve to perform random (but heralded) permutations to select which of the three pairwise overlaps $r_{AB}, r_{BC}$ and $r_{AC}$ is being measured in each experimental shot.}
    \label{fig:3swaps}
\end{figure}

To witness coherence, we need to obtain a triple of overlaps $\vec{r}=(r_{AB}, r_{BC}, r_{AC})$ that lies outside of the classical polyhedron $C$ described in \cref{fig:ditqudit}(b). In \cref{fig:regions}(a) we superimpose regions $Q$ and $C$ from \cref{fig:ditqudit}, and in \cref{fig:regions}(b) we show one illustrative two-dimensional cut of this space. In \cref{fig:regions}(b), region IV corresponds to overlaps that can be obtained from general quantum states (i.e.\ within the set $Q$), but which are not compatible with coherence-free triplets of states, i.e.\ do not satisfy the classical bounds (\ref{eq:rineq1})-(\ref{eq:rineq3}), being thus outside of $C$. 

\emph{Any} point that lies outside of $C$ but within $Q$ serves as a coherence witness, i.e., as a witness that there is no common basis in which the three states $\rho_{A}, \rho_{B}$ and $\rho_{C}$ are diagonal. In Appendix A we show that the maximal violation of inequality (\ref{eq:rineq1}) is obtained by three 2-dimensional pure states on a great circle on the Bloch sphere, separated by consecutive angles of $\pi/3$ (with $A$ being the central one). For these states, $r_{BC}=1/4$, and yet $r_{AB}=r_{AC}=3/4$, so that $r_{BC}=1/4 < r_{AB}+r_{AC}-1= 1/2$. Maximal violations of inequalities (\ref{eq:rineq2}) and (\ref{eq:rineq3}) can be obtained by permuting indices $A,B,C$ in the previous example.

Finally, one could worry that our coherence test might be fooled by classical states which are correlated, i.\ e., not independent.
If the three events $A$, $B$, and $C$ are not drawn independently, then we cannot interpret the quantities $r_{\rho \sigma}$ as two-state overlaps $\tr{(\rho \sigma)}$. If we nonetheless interpret $r_{\rho \sigma}$ operationally as that which is estimated in a SWAP test between $\rho$ and $\sigma$, then they must still satisfy our classical inequalities if the global state is diagonal in some local basis. To see that, consider the following.
When determining overlaps of two classical, deterministic states $A$ and $B$ via the SWAP test (see Fig. \ref{fig:swap2}-a), the outcome can only be either $p(0)=1$ (if states $A$ and $B$ are the same), or $p(0)=1/2$ (for different states). For overlaps $r_{AB}, r_{BC}, r_{AC}$ of three deterministic classical states $A, B$ and $C$, it is easy to convince oneself that the only possible deterministic values for the overlaps are: $(0,0,0)$ (all states different), $(1,1,1)$ (all states identical), and $(0,0,1), (0,1,0), (1,0,0)$ (one pair of identical states, with a third different from them). These are exactly the five vertices of the classical polyhedron $C$, whose faces are given by our classicality inequalities. Arbitrary convex combinations of deterministic classical states must result in points in the interior of $C$, so even if the states are classically correlated, the quantities $r_{ij}$, estimated via SWAP tests, cannot violate the classicality inequalities (\ref{eq:rineq1})-(\ref{eq:rineq3}).

\subsection{Characterizing multiphoton indistinguishability} \label{sec:multi}

In photonic experiments, a spectral function $\ket{\psi}$ is used to describe the degrees of freedom of a single photon which are inaccessible to the detectors (see e.g.\ \cite{Tichy15}). For two photons with spectral functions $\ket{A}$ and $\ket{B}$ entering different input ports of a 50/50 beam splitter, the probability of leaving the beam splitter in separate ports is $p=(1+r_{AB})/2$, in what is known as the Hong-Ou-Mandel (HOM) efffect \cite{MandelHO87}. The similarity of this expression with the outcome of a SWAP test is not accidental -- the HOM effect can be understood as a photonic realization of a SWAP test \cite{Filip02,Garcia-EscartinCP13}. For applications such as linear-optical quantum computation, it is necessary to have a high degree of indistinguishability between all photons whose worldlines cross. The traditional prescription to characterise multiphoton indistinguishability is to perform one HOM test for each photon pair, thus estimating all pairwise overlaps between the spectral functions $\ket{\psi_i}$. As we now describe, our results allow us to characterize multiphoton indistinguishability with a much smaller number of HOM tests.

If the $N$-vertex graph $G$ describing the $N$ photon states is connected, it has at least $N-1$ edges. For sufficiently high values of these known two-state overlaps, it is possible to use our results to obtain lower bounds on all unmeasured two-state overlaps. This reduces the number of required two-photon HOM tests from $O(N^2)$ to $O(N)$, which considerably simplifies the experimental effort required to characterize multiphoton sources. The simplest illustration of this advantage involves measuring two of the three possible overlaps among three nearly identical photons, and determining the bounds on the third overlap. In Fig. \ref{fig:ditqudit}(a) we see the region $Q$ of quantum-mechanically allowed triples of overlaps $\vec{r}=(r_{AB}, r_{BC}, r_{AC})$. If two of those overlaps are close to 1, it is clear that the third one also needs to be close to 1; our results provide precise bounds on the value of the third overlap.

Moreover, it is possible to construct a single multimode interferometer that measures any set of two-photon overlaps, as we now describe. Each photon is represented by a vertex in graph $G$, with edges representing the overlaps to be measured via HOM tests. Each photon enters a fan-out multimode interferometer with a number of modes that equals the degree of the vertex that represents it. These fan-out interferometers must connect the input port to all output ports; a balanced interferometer, such as the discrete Fourier transformation, is sufficient. Each output port of each fan-out interferometer undergoes a HOM test with an output port of another fan-out interferometer, according to the set of two-photon HOM tests specified by the edges of $G$. This general recipe results in a single multimode interferometer that performs the HOM tests specified by any graph $G$. By measuring the probability of bunching at each output port, we can characterize all two-photon overlaps in graph $G$, using our results to bound the unmeasured two-photon overlaps. In Figure \ref{fig:fan} we show an example of this general construction, featuring a 4-photon graph $G$ and its associated multimode interferometer.

\begin{figure}[ht]
    \centering
    \includegraphics[width=0.40\textwidth]{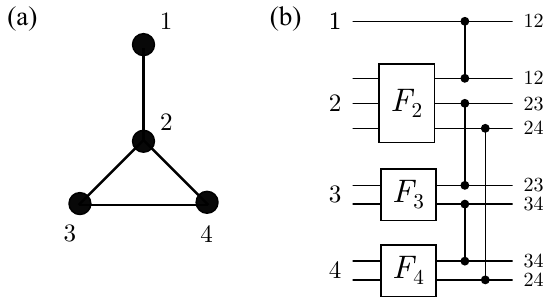}
    \caption{(a) Example of 4-vertex graph $G$ (known as a 4-vertex paw). Vertices represent photons, edges represent overlap measurements to be obtained from Hong-Ou-Mandel (HOM) tests. (b) Multimode interferometer that simultaneously measures all overlaps. Each fan-out Fourier interferometer $F_i$ distributes the input photon uniformly over its output ports. Then multiple two-mode HOM tests are performed (each pair of connected modes represents a 50/50 beam splitter). Indices at each final output port indicate which overlap will be measured using the probability of bunching in that port.}
    \label{fig:fan}
\end{figure}

This procedure gives a constructive way of measuring any set of two-photon overlaps, for any number of photons. This is a generalization of the interferometers we described in \cite{BrodGVFSS19} and in the recent preprint \cite{Giordani19b}, which correspond to the particular cases of $N$-vertex star graphs and the 4-vertex linear chain, respectively.

The HOM effect also illustrates an important aspect of our results, namely, that the applications do not require knowledge of all degrees of freedom of the system. The partial distinguishability observed between two photons can depend on  several unmeasured, inaccessible, and potentially continuous degrees of freedom, but nonetheless the outcome of the HOM effect gives an estimate of $r_{AB}$ suitable for all applications discussed here.

\subsection{Dimension witnesses} \label{sec:dim}

The bounds on overlaps may be different if we introduce constraints on the dimension of the Hilbert space spanned by the states. This means that overlap measurements can be used to certify that the states span a certain minimum Hilbert space dimension, which may have relevance to attacks on quantum cryptographic key distribution. Let us illustrate this application with a simple scenario.

Let us again consider the situation of three pure quantum states with two known overlaps $r_{AB}, r_{AC}$. In the most general situation we considered in section \ref{sec:b3states} and Appendix \ref{apx:overlaps}, the three states may span a 3-dimensional Hilbert space. In Appendix \ref{apx:qubits} we assume instead that the Hilbert space dimension spanned is only 2-dimensional, obtaining the bounds:
\begin{align}
r_{BC} &\le \left( \sqrt{r_{AB}r_{AC}}+\sqrt{(1-r_{AB})(1-r_{AC})}\right)^2 ,\label{eq:ubqb}\\
r_{BC} &\ge \left( \sqrt{r_{AB}r_{AC}}-\sqrt{(1-r_{AB})(1-r_{AC})}\right)^2. \label{eq:lbqb}
\end{align}
The upper bound (\ref{eq:ubqb}) coincides with that of the general case. The lower bound (\ref{eq:lbqb}), on the other hand, coincides with that of the general case only when $r_{AB}+r_{AC} \ge 1$; when $r_{AB}+r_{AC}<1$, the general case predicts a trivial lower bound ($r_{BC} \ge 0$). Thus, when $r_{AB}+r_{AC} < 1$, a measured value of $r_{BC}$ that violates inequality (\ref{eq:lbqb}) indicates that the Hilbert space spanned by the three states is necessarily  3-dimensional. As an example, any set of overlaps in region I of Fig.\  \ref{fig:regions}(b) is a dimension witness.

Equation (\ref{eq:implic3s}) is a single inequality that defines the convex body $Q$ of three overlaps achievable by three quantum states spanning, in general, 3 dimensions. The restriction to a 2-dimensional spanning Hilbert space corresponds to further imposing the constraint:
\begin{equation}
    r_{AB}+r_{BC}+r_{AC}+2\sqrt{r_{AB}r_{BC}r_{AC}}\geq1 .
\end{equation}

We emphasize that our dimension witness does not pertain to the dimension of the physical Hilbert space of the system, but rather to the dimension of the space spanned by the three states. It is possible to generate e.g.\ three photonic states which are defined on an infinite-dimensional Hilbert space, but which are not linearly independent vectors in this space. As discussed previously, in that case the HOM test gives a direct estimate of the overlaps, and allows us to rule out the possibility that the states span a 2-dimensional Hilbert space even without knowledge of which degrees of freedom are being physically used to encode the states.

Finally, we point out that classical dimension witnesses can also be found, i.e.\ for the coherence-free, diagonal states discussed in section \ref{sec:cbounds}. For the case of three classical bits, it is easy to show that the classical polyhedron of allowed overlaps is the tetrahedron defined as the convex hull of vertices $(r_{AB}, r_{AC}, r_{BC})=(0,0,1), (0,1,0), (1,0,0)$ and $(1,1,1)$. In comparison with the general classical polyhedron $C$ of Fig.\ \ref{fig:ditqudit}(b),  vertex $(0,0,0)$ is missing, due to the impossibility of finding three mutually orthogonal states of a single bit. Any overlap triple outside of the tetrahedron cannot correspond to diagonal, coherence-free states spanning two dimensions only.

\section{Conclusion}
\label{sec:conclusion}
We have considered the problem of obtaining constraints on the possible values of  pairwise quantum state overlaps among a set of quantum states. We completely solved the case of $N=3$ pure states, discussing how some sets of states violate the expectations of classical, coherence-free models, whose (linear) constraints we have also described. We have also discussed how to leverage these results to obtain bounds for the general scenario involving any number of states.

Our classicality inequalities derive from logical coherence, mirroring arguments used to derive non-contextuality and Bell non-locality inequalities \cite{Pitowsky94}. It would be interesting to obtain more concrete connections between our results and these two other notions of classicality. As quantum coherence is required to violate our classicality inequalities, our results may have an interpretation in terms of resource theories for quantum coherence \cite{StreltsovAP17}. In epistemic models of quantum theory \cite{BarrettCLM14, Leifer14, RingbauerDBCWF15}, the overlap between quantum states is partially accounted for by overlapping probability measures over the more fundamental ontic states; our classicality inequalities may prove helpful in quantifying the extent to which this is possible.

We have described three applications of our results: a novel type of coherence witness, dimension witnesses, and efficient characterisation of multiphoton indistinguishability. Our results may also find applications in other state-comparison protocols, such as quantum fingerprinting \cite{BuhrmanCWdW01}, and problems in quantum communication complexity \cite{BuhrmanCMdW10} and quantum machine learning \cite{BiamonteWPRWL17}.

\begin{acknowledgments}
We acknowledge support from CNPq project Instituto Nacional de Ci{\^e}ncia e Tecnologia de Informa{\c c}{\~a}o Qu{\^a}ntica. 
\end{acknowledgments}

\bibliography{ernesto_large}
\bibliographystyle{apsrev}

\onecolumngrid
\appendix

\section{Proof of quantum bounds on overlaps for three pure states}\label{apx:overlaps}

In this Appendix we deal with the case of three pure states with two given overlaps $r_{AB}=\ovlap{A}{B}$ and $r_{AC}=\ovlap{A}{C}$, and one unknown overlap $r_{BC}=\ovlap{B}{C}$, for which we find tight bounds, as described in the main text. In subsection \ref{apx:apgen} below we deal with the case of three pure states of any dimensionality; in section \ref{apx:purequbits} we solve the problem with the constraint that the three states that span only a 2-dimensional Hilbert space; and in section \ref{apx:maxviol} we find states yielding maximal violation of our classicality inequalities for this scenario [inequalities (\ref{eq:lci1})-(\ref{eq:lci3}) of the main text].

\subsection{General bounds on \texorpdfstring{$\ovlapr{B}{C}$}{}}
\label{apx:apgen}
Suppose we have three pure states $\ket{A}$, $\ket{B}$, and $\ket{C}$, such that both $\ovlapr{A}{B}$ and $\ovlapr{A}{C}$ are known. We begin by parameterizing these state in an economical way.

Without loss of generality, assume that the three states are spanned by basis states $\{\ket{0},\ket{1},\ket{2}\}$. Furthermore, we can align the basis states in a convenient way to use the following parameterization:
\begin{subequations}
\begin{align*}
\ket{A} &= \ket{0} \\
\ket{B} &= \cos{\beta} \ket{0} + e^{i a} \sin{\beta} \ket{1} \\
\ket{C} &= \cos{\gamma} \ket{0} + e^{i b}\sin{\gamma} \sin{\alpha} \ket{1}+e^{i c}\sin{\gamma} \cos{\alpha} \ket{2}.
\end{align*}
\end{subequations}
We can now apply a unitary transformation $U = \textrm{diag}(1,e^{-ia},e^{-ic})$ to all states and redefine $\phi = b-a$. This gives us three new states that have the same pairwise overlaps as the original ones, and we reach our final parameterization:
\begin{subequations}
\begin{align*}
\ket{A} &= \ket{0} \\
\ket{B} &= \cos{\beta} \ket{0} + \sin{\beta} \ket{1} \\
\ket{C} &= \cos{\gamma} \ket{0} + e^{i \phi}\sin{\gamma} \sin{\alpha} \ket{1}+\sin{\gamma} \cos{\alpha} \ket{2},
\end{align*}
\end{subequations}
where $\alpha$, $\beta$, $\gamma$  $\in [0,\pi/2]$ and $\phi \in [0, 2\pi)$. From this we have the two \emph{known} overlaps
\begin{subequations}
\begin{align*}
    \ovlapr{A}{B} & = \cos^2\beta, \\
    \ovlapr{A}{C} & = \cos^2\gamma,
\end{align*}
\end{subequations}
which are fixed, and the third overlap we wish to bound:
\begin{align} 
    \ovlapr{B}{C} =& \cos^2 \gamma \cos^2 \beta +\sin^2\gamma\sin^2\beta \sin^2\alpha +2 \sin\gamma\cos\gamma\sin\beta\cos\beta\sin\alpha\cos\phi.\label{eq:toopt}
\end{align}
Our goal now is to find the extrema of \cref{eq:toopt} with respect to parameters $\alpha$ and $\phi$. To that end we write
\begin{subequations}
\begin{align}
    \frac{\partial  \ovlapr{B}{C}}{\partial \alpha} & = 2 \sin^2 \beta \sin^2 \gamma \sin \alpha \cos \alpha  + 2 \sin \beta \cos \beta \sin \gamma \cos \gamma \cos \alpha \cos \phi \notag \\ &= 0 \label{eq:extr1}\\
    \frac{\partial  \ovlapr{B}{C}}{\partial \phi} & = - 2 \sin \beta \cos \beta \sin \gamma \cos \gamma \sin \alpha \sin \phi \notag \\&= 0 \label{eq:extr2}
\end{align}
\end{subequations}
Let us now break down all possible solutions of the above equations. Consider first \cref{eq:extr2}. It is true if either $\sin \alpha  =0$ or $\sin \phi = 0$.

If $\sin \alpha  =0$  then $\ket{C}$ has no support on $\ket{1}$, and the overlap between $\ket{B}$ and $\ket{C}$ reduces to
\begin{equation*}
     \ovlapr{B}{C} = \cos^2 \beta \cos^2 \gamma,
\end{equation*}
which is fixed by the two known overlaps.

If $\sin \phi = 0$, we have that naturally $\cos \phi = \pm 1$, and \cref{eq:extr1} reduces to
\begin{equation*}
    (\sin \beta \sin \gamma \sin\alpha \pm \cos \beta \cos \gamma) \cos \alpha = 0.
\end{equation*}
This equation has two solutions. The first is $\cos \alpha = 0$, in which case $\sin \alpha = 1$ (recall that $\alpha \in [0,\pi/2]$) and we have 
\begin{equation*}
     \ovlapr{B}{C} = \cos^2 (\beta \mp  \gamma).
\end{equation*}
The other solution occurs when 
\begin{equation*}
     \sin \alpha = \mp\frac{\cos \beta \cos \gamma}{\sin \beta \sin \gamma},
\end{equation*}
in which case we have
\begin{equation*}
    \ovlapr{B}{C} = 0.
\end{equation*}

We have thus obtained four extrema of $\ovlapr{B}{C}$: 
\begin{equation}
  \begin{cases}
    \cos^2 \beta \cos^2 \gamma,& \\
    \cos^2 (\beta \pm \gamma), &\\
    0,              & \text{if } \sin \alpha = \mp\frac{\cos \beta \cos \gamma}{\sin \beta \sin \gamma}
\end{cases}  
\end{equation}
We now need to check whether each of these values are maxima, minima or saddle points. The value 0 clearly is a minimum, but we also need to check under which conditions it can happen. As we show shortly, this minimum is attainable if
\begin{equation}
    \ovlapr{A}{B} + \ovlapr{A}{C} = \cos^2 \beta + \cos^2 \gamma >1.
\end{equation}
Interestingly, this is the same condition as that necessary to guarantee a nontrivial bound for overlaps of classical probability distributions, or coherence-free states [cf.\ Equation\ (\ref{eq:rineq1})]. In other words: although the general lower bound is looser than the lower bound for coherence-free states, the condition that guarantees they are nonzero is the same for both.

To investigate the extrema of $\ovlapr{B}{C}$ we use the following:

\begin{lemma}\label{lem:lemma1}
If $x, y \in (0,\pi/2)$ are such that $\cos^2 x + \cos^2 y > 1$, then the following hold:
\begin{align}
    \frac{\cos x \cos y}{\sin x \sin y}  &> 1 \\
    \cos^2 x \cos^2 y & \in [\cos^2(x+y), \cos^2(x-y)]
\end{align}
\end{lemma}
\begin{proof}
For the first part, write 
\begin{align*}
    (\tan x \tan y)^2 =& \frac{(1-\cos^2 x)(1-\cos^2 y)}{\cos^2 x \cos^2 y} \\
     =& \frac{1 - (\cos^2 x + \cos^2 y)}{\cos^2 x \cos^2 y} + 1 \\ & < 1,
\end{align*}
from which the inequality follows. For the second part, write
\begin{align*}
    [\cos (x\pm y)]^2 = &\cos^2 x \cos^2 y + \sin^2 x \sin^2 y \mp 2 \sin x \sin y \cos x \cos y \\
    = &\cos^2 x \cos^2 y  + \sin^2 x \sin^2 y \left(1 \mp \frac{2}{\tan x \tan y}\right)
\end{align*}
using the first inequality we see that the terms in parenthesis has the same sign as the plus/minus sign within, and so the second claim follows.
\end{proof}

We now set $x = \beta$ and $y = \gamma$ in \cref{lem:lemma1}, to conclude the following: whenever $\ovlapr{A}{B} + \ovlapr{A}{C} > 1$, the minimum $\ovlapr{B}{C} = 0$ does not occur. Furthermore, in this case the extremum given by $\cos^2 \beta \cos^2 \gamma$ is contained between the two values of $\cos^2 (\beta \pm \gamma)$, from which we conclude it must be a saddle point. Finally, combining all these together we conclude that, whenever 
\begin{equation*}
    \ovlapr{A}{B} + \ovlapr{A}{C} > 1,
\end{equation*}
the lower and upper bounds for $\ovlapr{B}{C}$ are $\cos^2 (\beta \pm \gamma)$. When $\ovlapr{A}{B} + \ovlapr{A}{C} \leq 1$, the upper bound for $\ovlapr{B}{C}$ is the same but the lower bound is 0. 

It is important to emphasize that, since these bounds were obtained by direct minimization over the free parameters $\alpha$ and $\phi$, they are always attainable. That is, given the two fixed overlaps, there always exist states $\ket{A}$, $\ket{B}$, and $\ket{C}$ for which $\ovlapr{B}{C}$ achieves the lower and the upper bounds. 

\subsection{Bounds for the case of three states spanning a 2-dimensional Hilbert space} \label{apx:purequbits}

Suppose now that the three states only span a 2-dimensional Hilbert space. For example, this could happen if the three systems happen to be qubits. Our general parameterization of the three states can now be written as

\begin{subequations}
\begin{align*}
\ket{A} &= \ket{0} \\
\ket{B} &= \cos{\beta} \ket{0} + \sin{\beta} \ket{1} \\
\ket{C} &= \cos{\gamma} \ket{0} + e^{i \phi}\sin{\gamma} \ket{1},
\end{align*}
\end{subequations}
where $\beta$, $\gamma$  $\in [0,\pi/2]$ and $\phi \in [0, 2\pi)$. From this it follows that
\begin{align*}
    \ovlapr{B}{C} = \cos^2 \beta \cos^2 \gamma +\sin^2\beta\sin^2\gamma +  2 \sin\beta\cos\beta\sin\gamma\cos\gamma\cos\phi.
\end{align*}
Differentiating with respect to $\phi$ to obtain the extrema we find
\begin{equation*}
    \frac{\partial  \ovlapr{B}{C}}{\partial \phi} = - 2 \sin \beta \cos \beta \sin \gamma \cos \gamma  \sin \phi = 0.
\end{equation*}
We conclude that the two extrema are
\begin{align*}
& \cos^2 \beta \cos^2 \gamma +\sin^2\beta\sin^2\gamma \pm 2 \sin\beta\cos\beta\sin\gamma\cos\gamma \notag = \cos^2(\alpha\pm \epsilon).
\end{align*}
Although this expression is similar to the two bounds found in the general case, the analysis here can be qualitatively different due to the possibility that $r_{AB}+r_{AC}\le 1$, in which case the lower bound for the qudit case is 0 in contrast with the qubit case.

\subsection{Maximal violation of classical bounds}
\label{apx:maxviol}
We now prove that the violation of the classical bound of 1/4 described in the text is the maximum possible. To prove this, we want to maximize the difference between the classical lower bound of Eq.\ (\ref{eq:rineq1}) and the quantum lower bound of Eq.\ (\ref{eq:lb}). In our parameterization, this difference can be written as
\begin{equation*}
D = \cos^2 \beta + \cos^2 \gamma -1 - \cos^2 (\beta + \gamma).
\end{equation*}
Notice that we assuming are assuming $\cos^2 \beta + \cos^2 \gamma >1$, otherwise both classical and quantum lower bounds become trivial. We now wish to maximize $D$ with respect to both $\beta$ and $\gamma$. To do this, we need
\begin{align*}
\frac{\partial  D}{\partial \beta} &= -2 \cos \beta \sin \beta + 2 \cos(\beta + \gamma)\sin(\beta +\gamma)   = 0,\\
\frac{\partial  D}{\partial \gamma} &= -2 \cos \gamma \sin \gamma + 2 \cos(\beta + \gamma)\sin(\beta +\gamma)=0.
\end{align*}
Simple manipulations show this is equivalent to
\begin{align*}
    \sin \gamma \cos (2 \beta + \gamma) &= 0, \\
    \sin \beta \cos (\beta + 2\gamma) &= 0.
\end{align*}
Recall that $\gamma, \beta \in [0, \pi/2]$. In this range, the solutions to these equations with $\sin \gamma = 0$ or $\sin \beta = 0$ are minima, since they lead to $D = 0$. The remaining solutions correspond to 
\begin{align*}
    \cos (2 \beta + \gamma) &= 0, \\
    \cos (\beta + 2\gamma) &= 0.
\end{align*}
In the domain of interest, these equations have a few solutions. By enumerating them it is easy to check that the  maximum of $D$ is 1/4. This maximum is obtained, for example, for $\gamma = \beta = \pi/3$. A set of three states which has these values for and reaches the maximal violation of the classical bound is
\begin{align*}
    \ket{A} & = \ket{0}, \\
    \ket{B} & = \tfrac{1}{2} \left(\ket{0} + \sqrt{3} \ket{1}\right), \\
    \ket{C} & = \tfrac{1}{2} \left(\ket{0} - \sqrt{3} \ket{1}\right).
\end{align*}
These states, up to a rotation of the Bloch sphere, correspond to three states in the equator of the Bloch sphere separated by consecutive angles of $\pi/3$, with $\ket{A}$ in the center, as claimed in the main text.

A similar calculation shows that the maximal quantum violation of the classical \emph{upper} bound is also 1/4. A set of three states that achieve this is
\begin{align*}
    \ket{A} & = \ket{0}, \\
    \ket{B} & = \tfrac{1}{2} \left(\ket{0} + \sqrt{3} \ket{1}\right),\\
    \ket{C} & = \tfrac{1}{2} \left(\sqrt{3}\ket{0} + \ket{1}\right).
\end{align*}

Note that these also correspond to three states in a great circle of the Bloch sphere separated by consecutive angles of $\pi/3$, as in the case of the maximal violation of the lower bound, but now we have $\ket{B}$ in the middle. This corresponds to the observation, in the main text, that the classical bounds of (\ref{eq:lci1}-\ref{eq:lci3}) can be obtained by each other from a relabeling of indices $A$, $B$ and $C$.

\section{Quantum bounds for mixed qubit states} \label{apx:qubits}

In this Appendix we prove that our quantum bounds of Eqs.\ (\ref{eq:ub})--(\ref{eq:lb}) extend to arbitrary mixed states for \emph{qubits}. As in Appendix \ref{apx:purequbits}, the quantum bounds we obtain for mixed qubit states in what follows hold regardless of whether condition $r_{AB}+r_{AC} > 1$ is satisfied.

As a warm-up, suppose that qubit $A$ is in a pure state (say, $\ket{0}$). Now suppose $B$ and $C$ are in states $\rho$ and $\sigma$, respectively, parameterized by:
\begin{equation*}
    \rho = 
 \begin{pmatrix}
  \rho_0 & \rho_1  \\
  {\bar{\rho_1}} & 1-\rho_0 
 \end{pmatrix} \textrm{ and }
 \sigma = 
 \begin{pmatrix}
  \sigma_0 & \sigma_1  \\
  {\bar{\sigma_1}} & 1-\sigma_0 
 \end{pmatrix},
\end{equation*}
where $\rho_0, \sigma_0 \in [0,1]$. Clearly the conditions that $\rho$ and $\sigma$ have trace 1 and are Hermitian are already taken into account by the parameterization. For these to be proper mixed states we also need $\tr \rho^2 \leq 1$ and $\tr \sigma^2 \leq 1$, which can be written as
\begin{align}
    \left| \rho_1 \right|^2 &\leq \rho_0 (1-\rho_0), \label{eq:purity1}\\
    \left| \sigma_1 \right|^2 &\leq \sigma_0 (1-\sigma_0) \label{eq:purity2}.
\end{align}

With these parameterizations, we can write the overlaps as
\begin{align}
    r_{AB} &= \rho_0,  \label{eq:rabapx}\\
    r_{AC} &= \sigma_0,  \label{eq:racapx}\\
    r_{BC} &= \rho_0 \sigma_0 + (1-\rho_0)(1-\sigma_0) + \bar{\rho_1}\sigma_1 + \rho_1 \bar{\sigma_1}. \label{eq:rbcapx}
\end{align}
We want to show that $r_{-} \leq r_{BC} \leq r_{+}$, where [cf.\ Equations (\ref{eq:ub})--(\ref{eq:lb})]
\begin{align}
    r_{\pm} = & r_{AB}r_{AC} + (1-r_{AB})(1-r_{AC})  \pm 2 \sqrt{r_{AB}r_{AC} (1-r_{AB})(1-r_{AC})} \notag\\
     = & \rho_0 \sigma_0 + (1-\rho_0)(1-\sigma_0)  \pm 2  \sqrt{\rho_0 \sigma_0 (1-\rho_0)(1-\sigma_0)}
\end{align}
Comparing the above with Eq.\ (\ref{eq:rbcapx}) we see that proving the required bounds on $r_{BC}$ is equivalent to showing that
\begin{equation}
    \left| \rho_1 \bar{\sigma_1} + \sigma_1 \bar{\rho_1} \right| \leq 2 \sqrt{\rho_0 \sigma_0 (1-\rho_0)(1-\sigma_0)}.
\end{equation}
Note now that
\begin{equation}
\left| \rho_1 \bar{\sigma_1} + \sigma_1 \bar{\rho_1} \right| \leq 2 \left| \rho_1 \right| \left| \sigma_1 \right| \leq 2 \sqrt{\rho_0 (1-\rho_0) \sigma_0 (1-\sigma_0)},
\end{equation}
where the first inequality follows from the triangle inequality, and the second follows from Eqs.\ (\ref{eq:purity1})--(\ref{eq:purity2}). Since the inequality holds, this implies that $r_{-} \leq r_{BC} \leq r_{+}$, as claimed.

Let us now extend the above result to when qubit $A$ is in a mixed state as well. Let us work in the basis where the state of qubit $A$ is diagonal, and we parameterize it as
\begin{equation*}
    \psi = 
 \begin{pmatrix}
  \psi_0 & 0  \\
 0 & \psi_1 
 \end{pmatrix},
\end{equation*}
such that $\psi_0 + \psi_1 = 1$. In this case, we have that Eqs.\ (\ref{eq:rabapx})--(\ref{eq:rbcapx}) become
\begin{align}
    r_{AB} &= \psi_0 \rho_0 + \psi_1 (1-\rho_0), \label{eq:rabapx2} \\
    r_{AC} &= \psi_0 \sigma_0 +  \psi_1 (1-\sigma_0),  \label{eq:racapx2}\\
    r_{BC} &= \rho_0 \sigma_0 + (1-\rho_0)(1-\sigma_0) + \bar{\rho_1}\sigma_1 + \rho_1 \bar{\sigma_1}. \label{eq:rbcapx2}
\end{align}
Our goal is again to prove that, for arbitrary $\psi_0$ and $\psi_1$, the bounds $r_{-}(\psi_0,\psi_1) \leq r_{BC} \leq r_{+}(\psi_0,\psi_1)$ hold, where
\begin{align*}
    r_{\pm}(\psi_0,\psi_1) = & r_{AB}r_{AC}+ (1-r_{AB})(1-r_{AC})  \pm 2 \sqrt{r_{AB}r_{AC} (1-r_{AB})(1-r_{AC})}.
\end{align*}
We write the dependence of $r_{pm}$ on $\psi_0$ and $\psi_1$ explicitly, but omit this dependence from $r_{AB}$ and $r_{AC}$ for simplicity of notation (note that $r_{BC}$ does not depend on $\psi_0$ and $\psi_1$).

We proved that $r_{-}(\psi_0,\psi_1) \leq r_{BC} \leq r_{+}(\psi_0,\psi_1)$ holds for $A$ pure, or equivalently for both limits $\psi_0 = 1$ and $\psi_1 = 1$. The fact that the bounds hold for all $\psi_0$ and $\psi_1$ follows from their concavity/convexity. More specifically, define the functions
\begin{equation*}
    f_{\pm}(x,y) = (\sqrt{x y}\pm \sqrt{(1-x)(1-y)})^2
\end{equation*}

A function $f(x,y)$ is convex on a convex region of  $\mathbb{R}^2$ if and only if its Hessian matrix is positive semidefinite in the interior of that region \cite{Bertsekas}. By testing this property we can show that $f_{-}$ is convex and $f_{+}$ is concave in the region defined by $x, y \in (0,1)$. In particular this means that, for $a,b \in [0,1]$ such that $a+b = 1$, we have
\begin{align*}
    f_{+}(a x_1 + b x_2,a y_1 + b y_2) & \geq a f(x_1, y_1) + b f(x_2,y_2), \\
    f_{-}(a x_1 + b x_2,a y_1 + b y_2) & \leq a f(x_1, y_1) + b f(x_2,y_2)
\end{align*}
By choosing $a=\psi_0$, $b=\psi_1$, $x_1 = \rho_0$, $x_2 = 1-\rho_0$, $y_1 = \sigma_0$, and $y_2 = 1-\sigma_0$, the above inequalities imply
\begin{equation*}
    r_{+}(\psi_0,\psi_1)\geq \psi_0 r_{+}(1,0) + \psi_1 r_{+}(0,1) \geq \psi_0 r_{BC} + \psi_1 r_{BC} = r_{BC},
\end{equation*}
where the last inequality follows from our previous results for the case of pure $A$ and from the fact that $r_{BC}$ does not depend on $\psi_0$ and $\psi_1$. By combining the above with a similar reasoning based on the convexity of $r_{-}$ we obtain
\begin{equation*}
    r_{-}(\psi_0,\psi_1) \leq r_{BC} \leq r_{+}(\psi_0,\psi_1),
\end{equation*}
as desired.

\section{Proof of the classical bounds}\label{ap:classb}

In this Appendix we prove bounds for the joint probability of $N$ events. These were proven by George Boole \cite{Boole1854}, but we include a proof for completeness. We then use those results to obtain inequalities that must be satisfied by classical, coherence-free diagonal states [as in Eq.\ (\ref{eq:cs}) of the main text], with known pairwise overlaps described by any connected graph $G$. These general inequalities,  described in Eq. (\ref{eq:genineq}) of the main text, have as a particular case the classical bounds for the 3-vertex graph $P_3$, i.e.\ inequalities (\ref{eq:lci1}-\ref{eq:lci3}).

Consider $N$ logical propositions $a_1, a_2, \ldots, a_N$, and let $p(a_i)$ be the probability that proposition $a_i$ holds. Let $p(a_1 \wedge a_2 \wedge \ldots \wedge a_N)$ be the probability that the joint proposition holds, i.e.\  that all $\{a_i\}$ are simultaneously true. We now show that logical coherence implies simple linear inequalities that $p(a_1 \wedge a_2 \wedge \ldots \wedge a_N)$ must satisfy.

We start with the simplest case of $N=2$ propositions $
\{a_1, a_2\}$. Using $0$ for false and $1$ for true, let us write the truth table for the AND ($\wedge$) function:

\begin{table}[ht]
  \begin{center}
    \begin{tabular}{|c|c|c|}
    \hline
      \textbf{$a_1$} & \textbf{$a_2$} & \textbf{$a_1 \wedge a_2$}\\
      \hline \hline
     0 & 0 & 0\\ \hline
     0 & 1 & 0\\ \hline
     1 & 0 & 0\\ \hline
     1 & 1 & 1\\ \hline
    \end{tabular}
\end{center}
\caption{Truth table for AND ($\wedge$) function.}
\label{table:and2}
\end{table}
Let us interpret each row in the table above as a vector $\vec{p}$ in a 3-dimensional space of probabilities $\vec{p}=(p(a_1), p(a_2), p(a_1 \wedge a_2))$. Since the table contains all possible truth assignments for $a_1$ and $a_2$, the most general, logically coherent vector $\vec{p}$ must be a convex combination of the rows of Table \ref{table:and2}. In our case, the logical coherence conditions for $\{ p(a_1), p(a_2), p(a_1 \wedge a_2) \}$ are simply the faces of the tetrahedron whose vertices are the rows of Table \ref{table:and2}.  The four faces are described by inequalities:
\begin{align}
p(a_1 \wedge a_2) &\ge 0 ; \label{in1}\\ 
p(a_1 \wedge a_2) &\le p(a_1) ; \label{in2}\\ 
p(a_1 \wedge a_2) &\le p(a_2) ; \label{in3}\\ 
p(a_1 \wedge a_2) &\ge p(a_1)+p(a_2)-1. \label{in4}
\end{align}
Inequality (\ref{in1}) is trivial; inequalities (\ref{in2}) and (\ref{in3}) simply state that the conjunction of two events must not happen more often than each of them separately, whereas inequality (\ref{in4}) gives a bound on $p(a_1 \wedge a_2)$, which follows from the inclusion-exclusion principle in probability theory.

The method described above, due to Pitowsky \cite{Pitowsky89,Pitowsky94}, is general, and can be applied to $m$ independent propositions together with any set of Boolean functions of them. First, we compile a list all $2^m$ truth values for the $m$ independent propositions, together with the corresponding truth values of the Boolean functions of interest (in the case above, $a_1 \wedge a_2$). The rows of the resulting table are then interpreted as vertices of a polytope, and its facets as our desired logical coherence conditions. These facets can be found using well-known convex hull algorithms (e.g.\ \cite{BarberDH96}).

We can apply the above method to $m$ propositions $a_1, a_2, ..., a_m$ and their joint proposition $a_1 \wedge a_2 \wedge ... \wedge a_m$. Each vertex of the polytope is a vector in a ($m+1$)-dimensional space of probabilities. Given the simplicity of the vertex list for this polytope, it is easy to check that the following inequalities are satisfied by all vertices, and hence by the complete polytope:
\begin{align}
p(a_1 \wedge a_2 \wedge ... \wedge a_m) &\ge 1-m+\sum_{i=1}^{m} p(a_i),\label{lcineq}\\
p(a_1 \wedge a_2 \wedge ... \wedge a_m) & \le p (a_i), \; \forall i = 1 \ldots m. \label{ineqlb}
\end{align}
Inequality (\ref{lcineq}) is saturated by exactly $m$ affinely independent vertices--those containing exactly two zeroes, plus the vertex with only ones--and thus constitutes a facet of the polytope. Each inequality (\ref{ineqlb}) is saturated by $2^{m-1}+1$ vertices, which also generate an $m$-dimensional face, i.e., a facet of the polytope. 

Let us now consider how to apply inequalities (\ref{lcineq}) and (\ref{ineqlb}) to obtain bounds for two-state overlaps of classical states. We start by considering $N$ independent random processes $A_i$, which yield outcomes $v(A_i)$ with probabilities $p[v(A_i)]$. Let $p_{ij}$ denote the probability that the independently drawn values for $A_i$ and $A_j$ are the same, so
\begin{equation*}
p_{ij} \equiv p[v(A_i)=v(A_j)]=\sum_k p[v(A_i)=k]p[v(A_j)=k],
\end{equation*}
where the sum is over all possible outcomes.

Consider an arbitrary, connected graph $G$ with $N$ vertices and $m$ edges. Each vertex represents a random process $A_i$, while edges $\{ i,j\} \in G$ represents a comparison between the outcomes of a pair of neighboring vertices/processes. We assign a logical proposition to each edge $\{i,j\} \in G$: 
\begin{equation}
a_{i,j} := v(A_i)=v(A_j), \forall \{i,j\} \in G.
\end{equation}
Inequality (\ref{lcineq}) then yields:
\begin{equation}
p\left[\bigwedge_{\{i,j\} \in G} v(A_i)=v(A_j) \right] \ge 1-m+\sum_{\{i,j\} \in G} p[v(A_i)=v(A_j)]
\end{equation}

Since $G$ is connected, for any vertex pair $\{k,l\}$ (even those not connected by edges of $G$), it is true that
\begin{equation}
p[v(A_k)=v(A_l)] \ge p\left[\bigwedge_{\{i,j\} \in G} v(A_i)=v(A_j) \right].  
\end{equation}
So
\begin{equation}
p[v(A_k)=v(A_l)] \ge 1-m+\sum_{\{i,j\} \in G} p[v(A_i)=v(A_j)], \forall \{k,l\}. \label{eq:mv}
\end{equation}

We now apply inequality (\ref{eq:mv}) above to obtain inequalities that bound the overlaps of classical states, defined as mixed states which are diagonal in a fixed, reference basis $\{\ket{\phi_i}\}$. As noted in the main text, the two-system overlap $r_{ij}=\tr(\rho_i \sigma_j)$ of classical states is the probability of obtaining the same outcome when measuring the two states in the classical basis:
\begin{align*}
\tr(\rho \sigma) =& \sum_i \bra{\phi_i} \rho \sigma \ket{\phi_i}\\
=&\sum_i \bra{\phi_i} \rho \ket{\phi_i}\bra{\phi_i}\sigma \ket{\phi_i}\\
=& \text{probability that } v(\hat{O})_{\rho}=v(\hat{O})_{\sigma}.
\end{align*}
Classical states can be viewed as a quantum way of parameterizing general independent probabilistic processes. This identification enables us to interpret inequality (\ref{eq:mv}) as an inequality about overlaps $r_{kl}=\tr(\rho_k \rho_l)$ of classical states, leading to
\begin{equation*}
r_{kl} \ge 1-m+\sum_{\{i,j\} \in G} r_{ij},
\end{equation*}
where $\{k,l\}$ are any pairs of vertices in $G$.
The inequality above actually represents many inequalities since, for any pair $\{k, l\}$, we can apply it to any connected subgraph of $G$ that contains these two vertices.
\end{document}